\pdfoutput=1
\documentclass{amsart}
\usepackage[utf8]{inputenc}
\usepackage[foot]{amsaddr}
\usepackage{fullpage}
\usepackage{mathabx}
\usepackage{amssymb}
\usepackage{bm}
\usepackage{mathrsfs}
\usepackage{physics}
\usepackage{upgreek}
\usepackage{booktabs}
\usepackage{multirow}
\usepackage{graphicx}
\usepackage{subcaption}
\usepackage{tikz}
\usepackage{algorithm}
\usepackage{algpseudocode}

\usetikzlibrary{automata}

\newtheorem{theorem}{Theorem}[section]
\theoremstyle{definition}

\DeclareMathOperator*{\minimize}{minimize}

\DeclareMathOperator*{\dist}{dist}

\newcommand{\CE}{\mathcal{E}}

\newcommand{\CV}{\mathcal{V}}

\newcommand{\SH}{\mathscr{H}}

\newcommand{\DR}{\mathbb{R}}
\newcommand{\DZ}{\mathbb{Z}}

\author{James Stokes$^{1\ast}$}
\address{$^1$Center for Computational Quantum Physics and Center for Computational Mathematics, Flatiron Institute, New York, NY 10010 USA}
\author{Saibal De$^{2\ast}$}
\author{Shravan Veerapaneni$^2$}
\address{$^2$Department of Mathematics, University of Michigan, Ann Arbor, MI 48109 USA}
\author{Giuseppe Carleo$^3$}
\address{$^3$Institute of Physics, \'{E}cole Polytechnique F\'{e}d\'{e}rale de Lausanne (EPFL), CH-1015 Lausanne, Switzerland}
\address{$^\ast$These authors contributed equally to this work.}

\title{Continuous-variable neural-network quantum states and the quantum rotor model}

\begin{document}

\begin{abstract}
    We initiate the study of neural-network quantum state algorithms for analyzing continuous-variable lattice quantum systems in first quantization. A simple family of continuous-variable trial wavefunctons is introduced which naturally generalizes the restricted Boltzmann machine (RBM) wavefunction introduced for analyzing quantum spin systems. By virtue of its simplicity, the same variational Monte Carlo training algorithms that have been developed for ground state determination and time evolution of spin systems have natural analogues in the continuum. We offer a proof of principle demonstration in the context of ground state determination of a stoquastic quantum rotor Hamiltonian. Results are compared against those obtained from partial differential equation (PDE) based scalable eigensolvers. This study serves as a benchmark against which future investigation of continuous-variable neural quantum states can be compared, and points to the need to consider deep network architectures and more sophisticated training algorithms.
\end{abstract}

\maketitle

\section{Introduction}

Variational Monte Carlo (VMC) approaches to the quantum many-body problem \cite{mcmillan1965ground} have witnessed a recent resurgence in activity fueled by the realization that when neural networks are exploited as systematically improvable trial wavefunctions, direct attack on otherwise intractable quantum many-body systems becomes possible. The success of so-called neural-network quantum states \cite{carleo2017solving} is closely paralleled by the ability of deep neural networks to overcome a related curse-of-dimensionality in a variety of high-dimensional machine learning tasks. The key features shared by these learning tasks is that they involve iterative fitting of a high-dimensional function approximator using various forms of stochastic approximation to an objective function, for which the learner has incomplete knowledge in the form of samples. Applying this perspective to the VMC for ground state determination, the variational wavefunction acts as a data generating process from which IID samples can be drawn. These samples provide noisy estimates for a Rayleigh quotient to be optimized, whose exact dependence on the variational parameters is unknown.

In this paper, we argue that quantum many-body simulation can leverage the success of geometric deep learning \cite{bronstein2017geometric, bronstein2021geometric} from two different perspectives which are based on group invariances of the Hamiltonian and the space of states, respectively. The first perspective applies whenever the configuration space of the Hamiltonian has a non-Euclidean structure, as in the hyperbolic lattices \cite{boettcher2020quantum, zhu2021quantum} currently under investigation using circuit quantum electrodynamics \cite{kollar2019hyperbolic}. In the second perspective, one is concerned with efficiently describing states of the Hilbert space which are invariant or equivariant to group invariances of the Hamiltonian. In the theory of many-body Schr\"odinger operators, for example, it can be shown that the unrestricted ground state inherits any group invariances of the parent Hamiltonian.

Motivated by the above simulation prospects, in this paper we initiate the study of continuous-variable VMC on non-Euclidean spaces, proposing a Hamiltonian which can be summarized by a small amount of geometric data, whose invariances are respected by the quantum theory. We limit our investigation to the simplest non-trivial target space geometry corresponding to a system of quantum rotors.  Ref.~ \cite{iblisdir2007matrix} undertook a variational ground state study of the quantum rotor model in one spatial dimension using basis truncation and matrix product states. The neural network approach advocated here, in contrast, does not rely on a choice of basis {and} could potentially be advantageous in the analysis of more complicated geometries such as hyperbolic manifolds, which are central to the study of quantum chaos \cite{gutzwiller1985geometry}. Another possible advantage of the neural network approach compared to \cite{iblisdir2007matrix} is the lack of reliance on the density matrix renormalization group training algorithm \cite{white1992density, white1993density}, which tends to struggle beyond one spatial dimension or in applications to disordered systems without a lattice.

Rather than pursuing the full machinery of geometric deep learning for variational simulation, we choose to focus on the introduction of baselines and illustration of general techniques using minimalist architectures which generalize those originally introduced for quantum spin systems \cite{carleo2017solving}. In particular, experiments are conducted using a rotor variant of the restricted Boltzmann machine applicable to the planar quantum rotor model, which can be understood as a continuous-variable relaxation of the quantum Ising model. The results are compared against scalable PDE-based eigensolvers and the extension of the method to other geometries is discussed.

The paper is structured as follows: we first introduce a geometrically motivated Hamiltonian, highlighting some subtleties that arise in non-Euclidean space. The planar quantum rotor model is identified as the simplest candidate system for illustrating the applicability of geometric machine learning techniques. Baselines are introduced for ground state preparation of the rotor model, which we investigate using a combination of variational techniques inspired by shallow neural networks, as well as techniques based on scalable PDE eigensolvers. We conclude by summarizing future directions.

\section{States and Hamiltonian}

The quantum systems under consideration describe finitely many particles moving on a Riemannian target space subject to two-body interactions. In particular, given a choice of Riemannian manifold $(M,g)$, a finite simple undirected graph $G = (\mathcal{V},\mathcal{E})$ decorated by vertex and edge weights functions $h : \mathcal{V} \to \mathbb{R}_{\geq 0}$ and $\beta : \mathcal{E} \to \mathbb{R}$ respectively, and a choice of potential function $V: \mathbb{R}_{\geq 0} \to \mathbb{R}$, we define the generator of time evolution $H$ in the infinite-dimensional Hilbert space of states $\SH = \bigotimes_{i \in \mathcal{V}} L^2(M)$ as follows,
\begin{equation}\label{e:hamiltonian}
    H\psi = -\frac{1}{2}\sum_{i \in \mathcal{V}} h_i \, \Delta_{g_i}\psi + \sum_{\{ i,j \} \in \mathcal{E}} \beta_{ij} \, V\big(\!\dist(x_i,x_j)\big) \psi \enspace ,
\end{equation}
where $\Delta_{g_i}$ denotes the Laplace-Beltrami operator acting on $M_i$ and $\dist : M \times M \to \mathbb{R}_{\geq 0}$ denotes the Riemannian distance function. Since the Hamiltonian $H$ depends only on the intrinsic geometry of the Riemannian manifold (via the distance and Laplace operator), the associated quantum theory inherits the invariances of the Riemannian space $(M,g)$ and likewise for the invariances of the interaction graph. Under appropriate assumptions on the potential, these invariances are inherited by the unique positive ground state of $H$ \cite{glimm2012quantum}. The geometry of the target space $(M,g)$ has non-trivial implications for the smoothness of the potential energy. In particular, the distance function $\dist(x,\cdot)$, viewed as a function of its second argument, will generically suffer from cusp singularities at $x$ and at any points $y \in M$ admitting multiple minimizing geodesics to $x$ (the so-called cut locus).

The primary goals of this paper are the development of scalable variational approaches to solve  both the ground state eigenvalue problem and the time evolution problem corresponding to the Schr\"{o}dinger operator \eqref{e:hamiltonian}. In this paper, however, we only consider the ground state eigenvalue problem, which can be rephrased as the following unconstrained optimization problem over $\SH$,
\begin{equation}
    \minimize_{\psi \in \SH} R_H(\psi), \quad R_H(\psi) = \frac{\langle \psi, H \psi \rangle}{\langle \psi, \psi \rangle} \enspace .
\end{equation}

The simplest Riemannian manifold after Euclidean space is the $d$-dimensional unit sphere $S^d$. In the interests of simplicity we focus in this initial work on $d=1$, corresponding to the unit circle $S^1$. In terms of the implicit angular parametrization $\theta \in [0,2\pi)$, the Riemannian distance function on the circle is given by
\begin{equation}
    \dist(\theta,\theta') = \min\{|\theta- \theta'|, 2\pi - |\theta - \theta'| \} \enspace ,
\end{equation}
which exhibits the expected cusp singularities at antipodal points defined by $|\theta - \theta'| \in \{0, \pi \}$. In this example the cut locus is a single point. The singularities of the distance function can smoothed out by a suitable choice of potential, which we take to be of cosinusoidal form. The resulting Hamiltonian is that of the quantum rotor model formulated on an arbitrary graph,
\begin{equation}
    \label{eq:rotor-hamiltonian}
    H = - \frac{1}{2} \sum_{i \in \CV} h_i \frac{\partial^2}{\partial \theta_i^2} + \sum_{\{ i,j \}\in \CE} \beta_{ij} \big[2-2\cos(\theta_i - \theta_j)\big] \enspace .
\end{equation}

The analysis of manifolds other than the sphere is also clearly of interest. In the Appendix \ref{sec:sigma} we explain how the Hamiltonian \eqref{e:hamiltonian} provides a lattice regularization for quantum non-linear sigma models.

\section{Rotor Restricted  Boltzmann Machine}

Inspired by the restricted Boltzamnn neural network quantum state introduced in \cite{carleo2017solving}, we introduce a class of trial wavefunctions suitable for time evolution and ground state determination for the quantum rotor model \eqref{eq:rotor-hamiltonian}. Since everything generalizes to rotors of arbitrary dimension, we first consider the $S^d$ target space and subsequently specialize to the circle ($d=1$). Denote the classical configuration of $n := |\mathcal{V}|$ visible rotors by $x := (\vec{x}_1,\ldots,\vec{x}_n) \in (S^d)^n$. Following \cite{carleo2017solving}, we assign a probability amplitude to each configuration of rotors $x \in (S^d)^n$ by integrating a Boltzmann factor over a space consisting of $m$ hidden rotors, whose collective coordinates are denoted $z := (\vec{z}_1,\ldots,\vec{z}_m) \in (S^d)^m$.
In order to construct a suitable Boltzmann factor, we consider the isometric embedding of the target space $S^d \subseteq \mathbb{R}^{d+1}$ into ($d+1$)-dimensional Euclidean space and choose the exponent to be of restricted Boltzmann form,
    \begin{equation}
    \psi(x) = \int_{(S^d)^m} {\rm d}\mu(z) \exp\left[\sum_{i=1}^m\sum_{j=1}^n a_{ij} \langle \vec{z}_i, \vec{x}_j \rangle  + \sum_{i=1}^m \langle \vec{b}_i, \vec{z}_i \rangle + \sum_{j=1}^n \langle \vec{c}_j, \vec{x}_j \rangle\right] \enspace ,
\end{equation}
where ${\rm d}\mu(z)$ denotes the surface measure on $(S^d)^m$ (the counting measure for $d=0$).
In the above expression $\langle \cdot, \cdot \rangle$ denotes the dot product for $\mathbb{R}^{d+1}$ and the variational parameters are given by  $a_{ij} \in \mathbb{R}$,  $\vec{b}_i \in \mathbb{R}^{d+1}$ and $\vec{c}_j \in \mathbb{R}^{d+1}$ for all $(i,j) \in [m] \times [n]$.
If all the bias terms vanish, then the amplitude is invariant under global $O(d+1)$ transformations of the visible rotors, the proof of which follows by a change of integration variables combined with invariance of the integration measure. For $d=0$ one has $S^0 = \{\pm 1\} \subseteq \mathbb{R}$ and $O(1) = \mathbb{Z}_2$, reproducing the proposal of \cite{carleo2017solving}. The weights and biases can also be promoted to complex numbers, resulting in a holomorphic parametrization  suitable for time evolution. In this paper, however, we only consider real parametrizations since the Hamiltonian \eqref{eq:rotor-hamiltonian} is known to have a non-negative ground state. The integration over the hidden rotors can be performed for any value of $d \in \mathbb{N}$. In the case of relevance to \eqref{eq:rotor-hamiltonian} we obtain the following logarithmic probability amplitude
\begin{equation}
    \log\psi(x) = \sum_{j=1}^n \langle \vec{c}_j, \vec{x}_j\rangle + \sum_{i=1}^m \log \big[2\pi I_0(\Vert \vec{y}_i \Vert)\big] \enspace ,
\end{equation}
where $y = ax + b$ is an affine transformation of the embedded rotor configurations $x \in \mathbb{R}^{n\times(d+1)}$, defined in terms of $a = (a_{ij})$ and $b \in \mathbb{R}^{n \times (d+1)}$, and $I_0$ denotes a modified Besssel function. The model can be trained using an efficient Markov Chain Monte Carlo method {generalizing \cite{carleo2017solving}, which is summarized in Appendix \ref{sec:training}}.

\section{Benchmarks: Jastrow Variational Wavefunction}

It is instructive to compare the form of the continuous-variable RBM with that of a $O(d+1)$-invariant Jastrow wavefunction, which is defined by local interactions dictated by the choice of interaction graph,
\begin{equation}\label{e:jastrow}
    \log\psi_{\rm J}(x) = \sum_{\{i,j\} \in \mathcal{E}} w_{ij} \langle \vec{x}_i, \vec{x}_j \rangle \enspace ,
\end{equation}
where $w_{ij}$ denote the variational parameters characterizing the trial wavefunction. Since the number of parameters is dictated by the choice of interaction graph, the Jastrow wavefunction, unlike the RBM, lacks the property of systematic improvability. The Jastrow wavefunction has the advantage that the Rayleigh quotient $R_H(\psi_J)$ can be computed analytically for certain interaction graphs. We carry out this calculation in the case of a linear (path) graph in Appendix \ref{sec:jastrow}

\section{Benchmarks: High-Dimensional PDE Solvers}

As a second validation of the ground states obtained from VMC simulations, we compare them against those obtained using partial differential equation (PDE) based eigensolvers. Here, we give a brief overview of our implementation of these PDE solvers. We restrict ourselves to the Hamiltonian (\ref{eq:rotor-hamiltonian}) associated with the unit circle $S^1$ with fixed positive weight $h_i = h$ on vertices $\CV = \qty{1, \ldots, n}$.
The associated eigenvalue problem is given by
\begin{equation}
    \label{eq:schrodinger}
    H \psi(\vb*{\theta}) = \lambda \psi(\vb*{\theta}) \enspace , \quad \quad H = - \frac{h}{2} \sum_{i = 1}^n \pdv[2]{\theta_i} + \sum_{\{ i,j \}\in \CE} \beta_{ij} \big[2-2\cos(\theta_i - \theta_j)\big] \enspace ,
\end{equation}
where we denote $\vb*{\theta} = (\theta_1, \ldots, \theta_n)$, $\theta_i \in [0, 2\pi)$ for $1 \leq i \leq n$.

One straightforward approach to solving \eqref{eq:schrodinger} is to discretize the domain using a regular grid on $[0, 2\pi)^n$, applying the finite difference approximation to the Schr\"odinger operator and solving the resulting algebraic eigenvalue problem. This will lead to a solution scheme whose error decays only polynomially depending on the order of the finite difference scheme used to approximate the Hamiltonian \cite{leveque2007finite}. Instead, we can take advantage of the periodicity of the domain and switch to the frequency domain using Fourier series expansion for the wavefunction:
\begin{equation}
    \label{eq:fourier}
    \psi(\vb*{\theta}) = \frac{1}{(2\pi)^n}\sum_{\vb*{\omega} \in \DZ^n} \hat{\psi}(\vb*{\omega}) \exp(-i \vb*{\omega} \cdot \vb*{\theta}) \enspace , \quad \quad \vb*{\omega} \cdot \vb*{\theta} = \sum_{i = 1}^n \omega_i \theta_i \enspace .
\end{equation}
As is well known, assuming $\psi$ is smooth, its truncated Fourier approximation will yield spectral accuracy. A numerical solver for \eqref{eq:schrodinger} can be constructed based on transforming the eigenvalue problem into the Fourier domain; its description is provided in Appendix \ref{sec:Fourier}, along with details on the eigenvalue algorithms.

One major obstacle to implementing either the finite difference or the Fourier spectral schemes described above is the memory cost: assuming $m$ degrees of freedom per dimension (number of grid points for finite difference or number of Fourier modes for spectral schemes), the discretized wavefunction will require storing $m^n$ scalar variables for a full representation. As the dimensionality of the problem (number of rotors) increases, it quickly becomes impossible to store the eigenfunction in the memory of a single computing node. This is a direct consequence of the \emph{curse of dimensionality}. To alleviate this issue and help scale our benchmark solvers, we adopt the distributed memory computing model, where we split the wavefunction across multiple computing nodes and use the message passing interface (MPI) library to implement necessary communication between the nodes. For our particular form of the Hamiltonian, the discretized operator ends up being sparse (both in the case of finite difference and Fourier spectral schemes), and this allows us to design fast matrix-vector products with minimal communication \cite{saad2003iterative}.

\section{Numerical Results}

We implemented the PDE eigensolvers and VMC algorithm in C++ and using state-of-the-art open source libraries; the PDE eigensolvers were built on top of Trilinos \cite{heroux2005overview} to support distributed computing. The code is available publicly at \texttt{https://github.com/shravanvn/cnqs}. The numerical experiments described in this section were run using version 1.0.0 of the code. All simulations were run on the Great Lakes cluster at the University of Michigan. Each compute node of this cluster is equipped with two 18-core 3.0 GHz Intel Xeon Gold 6154 processors and 192 GB RAM.

\subsection{Convergence of the Benchmark PDE Solver}

\begin{figure}
    \centering
    \begin{subfigure}{0.3\textwidth}
        \centering
        \begin{tikzpicture}
    \node [state] at (0, 2) (n1) {1};
    \node [state] at (2, 2) (n2) {2};
    \node [state] at (0, 0) (n3) {3};
    \node [state] at (2, 0) (n4) {4};
    
    \draw (n1) edge (n2);
    \draw (n1) edge (n3);
    \draw (n2) edge (n4);
    \draw (n3) edge (n4);
\end{tikzpicture}
        \caption{$2 \times 2$ lattice}
    \end{subfigure}
    ~
    \begin{subfigure}{0.3\textwidth}
        \centering
        \begin{tikzpicture}
    \node [state] at (0, 2) (n1) {1};
    \node [state] at (2, 2) (n2) {2};
    \node [state] at (0, 0) (n3) {3};
    \node [state] at (2, 0) (n4) {4};
    
    \draw (n1) edge (n2);
    \draw (n1) edge (n3);
    \draw (n1) edge (n4);
    \draw (n2) edge (n4);
    \draw (n2) edge (n3);
    \draw (n3) edge (n4);
\end{tikzpicture}
        \caption{4-node complete graph}
    \end{subfigure}
    \caption{Four-rotor networks used in convergence analysis of the PDE eigensolver.}
    \label{fig:four-rotor-networks}
\end{figure}
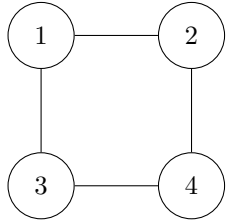
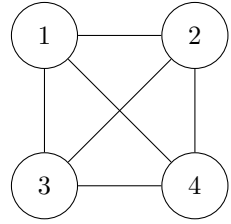

We performed self-convergence analysis of our PDE eigensolver on two four-rotor networks depicted in Figure~\ref{fig:four-rotor-networks} with vertex weights $h_i = 5$ and edge weights $\beta_{ij} = 1$. We ran the Fourier PDE eigensolver for maximum frequencies $1 \leq \omega_\text{max} \leq 32$, and Conjugate Gradient (CG) and inverse power iteration tolerances $\tau_\text{cg} = \tau_\text{inv} = 10^{-15}$.

\begin{figure}
    \begin{minipage}{0.48\textwidth}
        \centering
        \includegraphics[width=\linewidth]{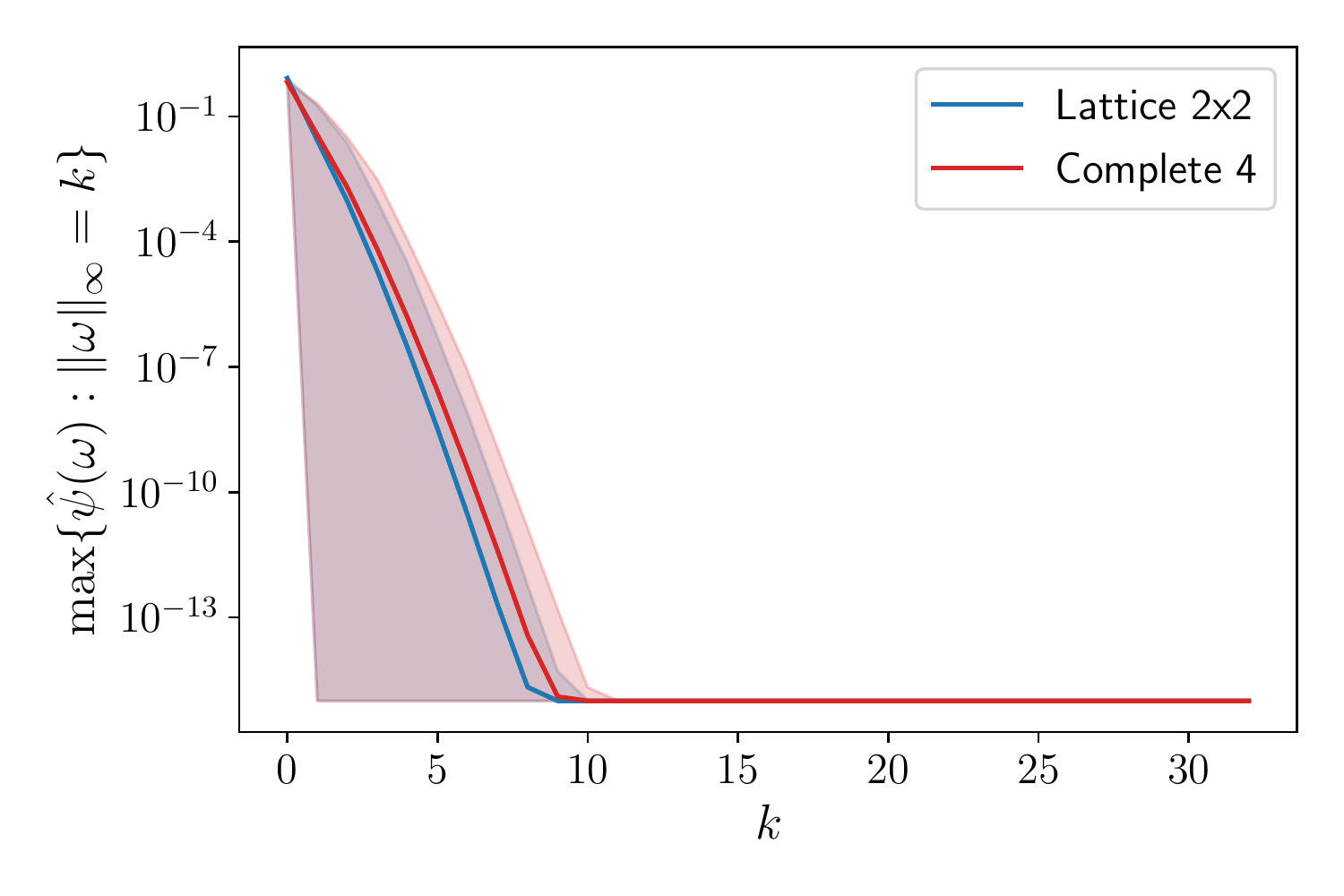}
        \caption{Decay of the ground state amplitudes with increasing frequency corresponding to the four-rotor networks computed from $\omega_\text{max} = 32$ discretization of the Fourier PDE eigensolver.}
        \label{fig:state-amplitude-decay}
    \end{minipage}
    ~
    \begin{minipage}{0.48\textwidth}
        \centering
        \includegraphics[width=\linewidth]{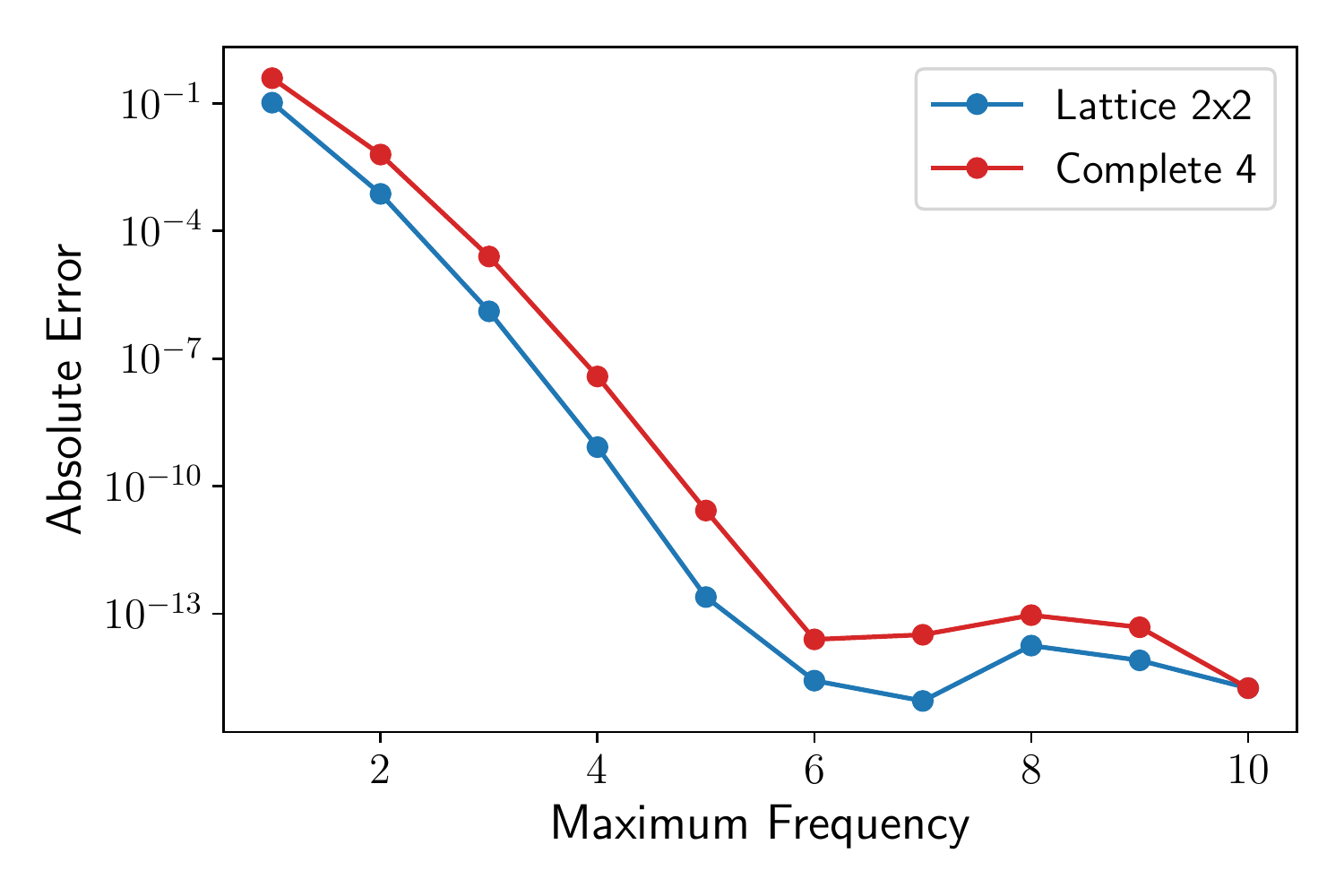}
        \caption{Spectral decay of error in the minimum eigenvalue estimate for the four-rotor networks with increasing maximum frequency. The error is estimated using the ground state computed from $\omega_\text{max} = 32$ discretization.}
        \label{fig:spectral-decay}
    \end{minipage}
\end{figure}

We plot the state amplitude for both networks corresponding to the $\omega_\text{max} = 32$ discretization in Figure~\ref{fig:state-amplitude-decay}; as we can see, the amplitudes vanish for approximately $\Vert \omega \Vert_\infty \geq 10$. This leads to spectral convergence of our Fourier PDE eigensolver; we demonstrate this in Figure~\ref{fig:spectral-decay}, where we plot the error in the ground state energy (using $\omega_\text{max} = 32$ solution as reference) as we increase the maximum frequency. The error decays exponentially until it reaches machine precision.

\begin{figure}
    \begin{minipage}{0.48\textwidth}
        \centering
        \includegraphics[width=\linewidth]{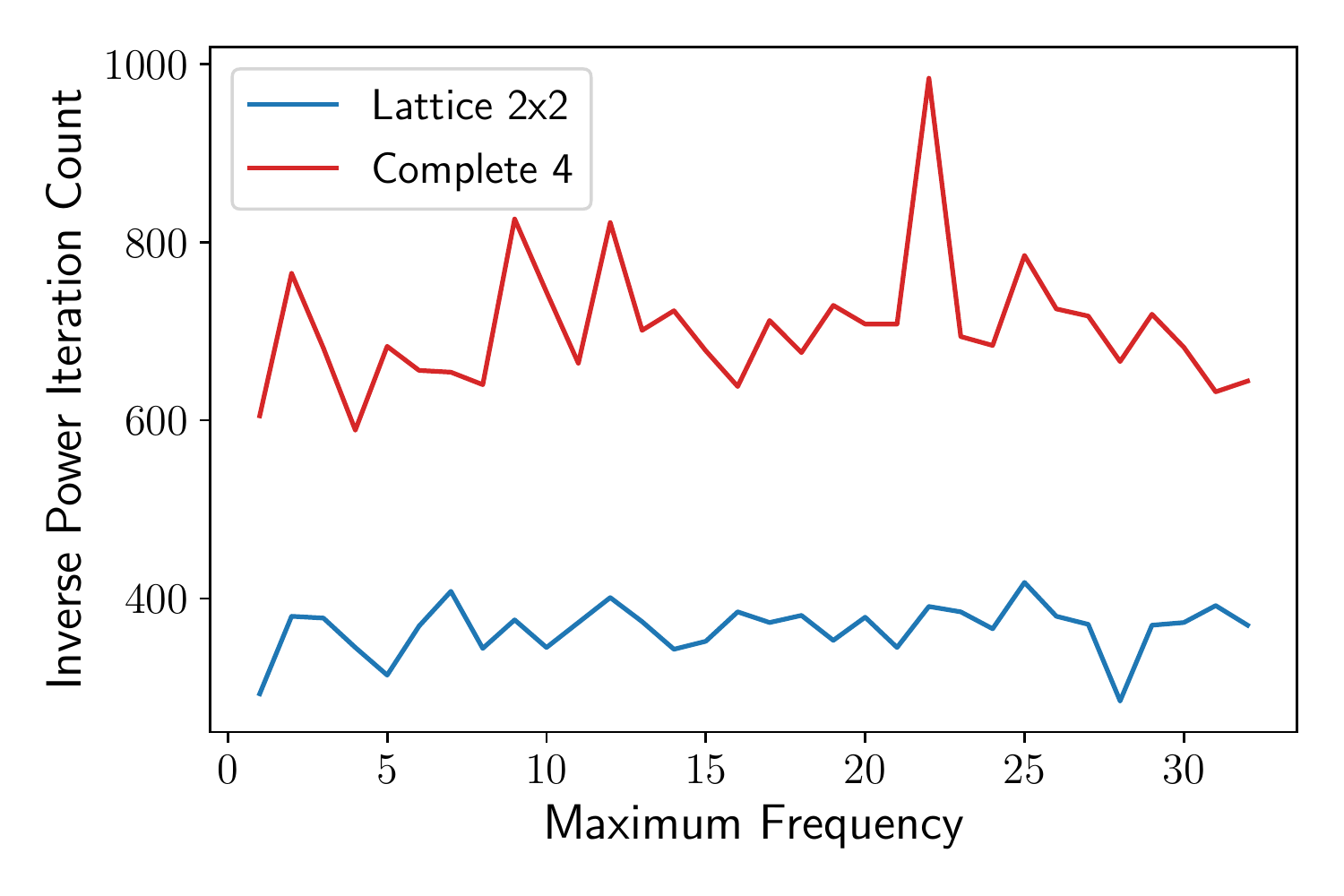}
        \caption{Number of inverse power iterations required before convergence.}
        \label{fig:num-inv-power-iter}
    \end{minipage}
    ~
    \begin{minipage}{0.48\textwidth}
        \centering
        \includegraphics[width=\linewidth]{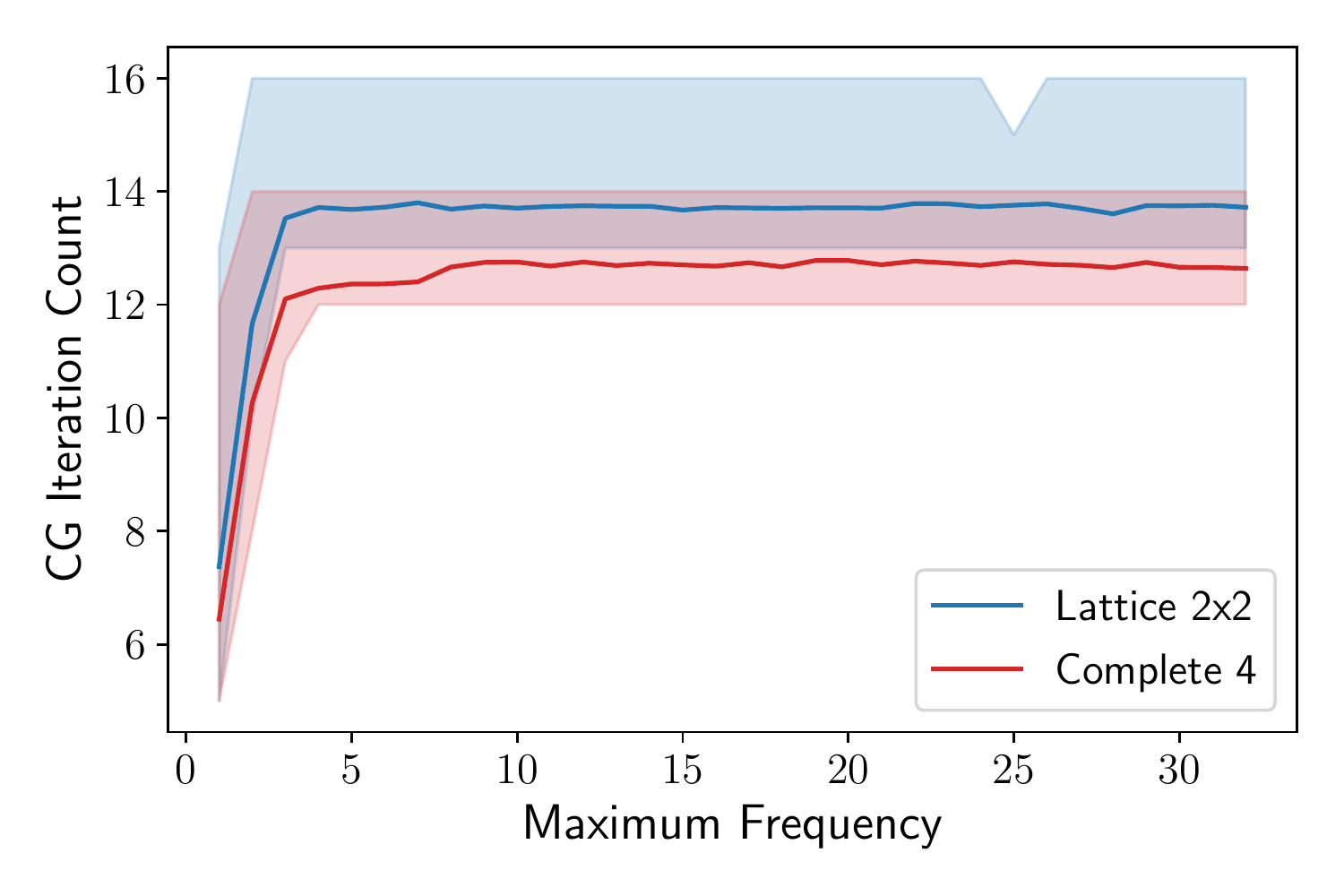}
        \caption{Average number of CG iterations per inverse power iteration; shaded region depict the range of the number of CG iterations.}
        \label{fig:num-cg-iter}
    \end{minipage}
\end{figure}

In Figure~\ref{fig:num-inv-power-iter} we plot the number of inverse power iterations necessary for convergence; we note that the iteration count is largely independent of the maximum frequency discretization parameter. Additionally, from Figure~\ref{fig:num-cg-iter} we see that with our preconditioner for the Fourier problem, the number of CG iterations per inverse iteration plateaus very quickly.

\subsection{VMC Simulations}

We ran VMC simulations with the same four-rotor Hamiltonians. In these simulations, 20 hidden nodes were used to construct the restricted Boltzmann machine and 10000 stochastic gradient descent steps were performed with a learning rate of $10^{-2}$. During each of these steps 24000 Metropolis-Hastings samples were generated, the first 4000 of which were discarded as burn-in and then every 20th sample were cherry-picked to compute the expected value quantities (e.g.\ energy and gradient). A stochastic reconfiguration parameter of $10^{-6}$ was selected for the simulations. { For details on the stochastic reconfiguration algorithm we refer the reader to the appendices of \cite{carleo2017solving}.}

\begin{figure}
    \centering
    \includegraphics[width=0.48\textwidth]{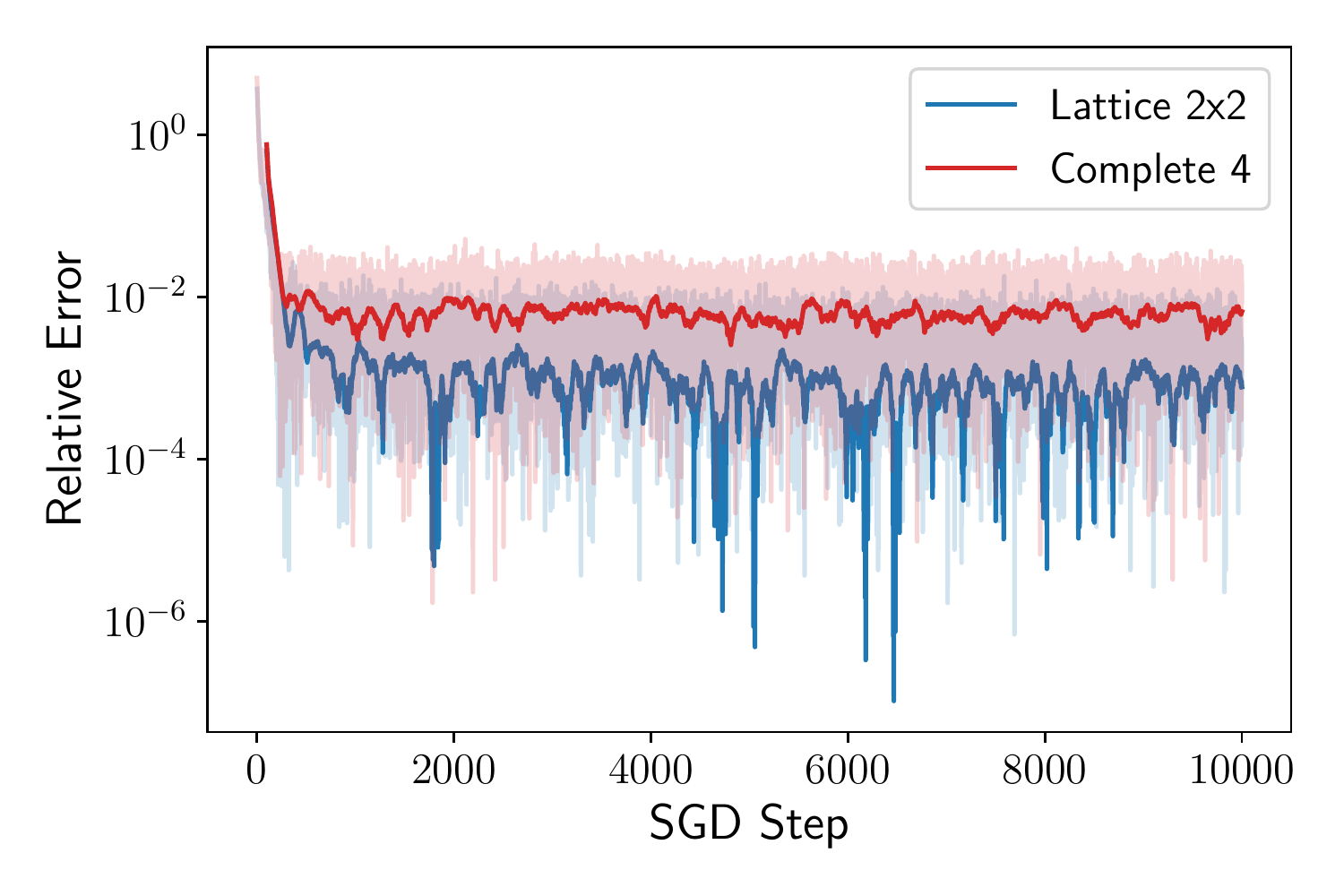}
    \caption{Evolution of relative error in the energy estimate (taking the energy computed using the PDE eigensolver as reference) during the stochastic gradient descent in the VMC simulations for four-rotor networks. The solid line represents the relative error when using 250-point rolling average to compute the VMC energy estimate.}
    \label{fig:four-rotor-vmc}
\end{figure}

In Figure~\ref{fig:four-rotor-vmc}, we demonstrate the convergence of the energy estimates for the ground state of the corresponding Hamiltonians over the optimization process.

\subsection{Comparison of PDE and VMC Solvers}

\begin{table}
  \renewcommand*{\arraystretch}{1.2}
  \centering
  \begin{tabular}{r r r r r r r r r}
    \toprule
    \multirow{2}{*}{\# Rotors ($n$)} & \multicolumn{2}{c}{PDE ($\omega_\text{max} = 5$)} & \multicolumn{2}{c}{PDE ($\omega_\text{max} = 7$)} & \multicolumn{3}{c}{VMC} & \multicolumn{1}{c}{Jastrow} \\
    \cmidrule(lr){2-3} \cmidrule(lr){4-5} \cmidrule(lr){6-8} \cmidrule(lr){9-9} & $\lambda_\text{min}$ & $T_\text{elapsed}$ & $\lambda_\text{min}$ & $T_\text{elapsed}$ & $\text{avg}$ & $\text{std}$ & $T_\text{elapsed}$ & $R_H(\psi_{\rm J})$ \\
    \midrule
    2 & $1.625$ &    $0.04$ & $1.625$ &    $0.04$ & $1.625$ & $1.2 \times 10^{-3}$ &  $364$ & 1.627 \\
    3 & $3.235$ &    $0.2$ & $3.235$ &    $0.2$ & $3.235$ & $2.1 \times 10^{-2}$ &  $721$ & 3.254 \\
    4 & $4.844$ &    $1$ & $4.844$ &    $4$ & $4.844$ & $4.4 \times 10^{-2}$ & $2029$ & 4.882 \\
    5 & $6.452$ &   $20$ & $6.452$ &  $116$ & $6.454$ & $6.6 \times 10^{-2}$ & $2667$ & 6.509 \\
    6 & $8.059$ &  $392$ & $8.059$ & $3768$ & $8.058$ & $8.9 \times 10^{-2}$ & $3934$ & 8.136 \\
    7 & $9.667$ &  $894$ & $9.667$ & $9579$ & $9.669$ & $1.0 \times 10^{-1}$ & $5573$ & 9.763 \\
    \bottomrule
  \end{tabular}
  \vspace{1em}
  \caption{Comparison of the PDE and VMC based solvers for the quantum rotor model one a linear graph. The PDE eigensolvers were run with CG and inverse power iteration tolerances $\tau_\text{cg} = \tau_\text{inv} = 10^{-15}$ and two discretization settings, $\omega_\text{max} = 5$ and $\omega_\text{max} = 7$, were used to ensure convergence to the smallest eigenvalue. 16 cores were utilized for the $n = 7$ rotor PDE simulations, as such the reported wall-time should be multiplied by a factor of 16 to obtain the actual CPU compute-time. The VMC eigensolvers were run for $10000$ stochastic gradient descent steps with learning rate $10^{-2}$, $5n$ hidden RBM nodes, $6000n$ Metropolis-Hastings samples per gradient descent step ($1000n$ initial samples were discarded as burn-in and then every $5n$ samples were cherry-picked). All elapsed times are in seconds. {For the VMC, avg and std refer to the average and standard deviation of the energy computed in the optimized RBM state.} The last column shows the optimal energy $1.62718 (n-1)$ of the Jastrow wavefunction \eqref{e:jastrow} assuming uniform weights $w_{ij}$.}
  \label{tab:pde-vmc}
\end{table}

In Table~\ref{tab:pde-vmc}, we compare the performances of the PDE and VMC solvers on linear chain networks. For each $2 \leq n \leq 7$, we ran inverse power iteration with tolerances $\tau_\text{inv} = 10^{-15}$ and $\tau_\text{cg} = 10^{-15}$ for the eigenproblem in Fourier domain with two discretizations: $\omega_\text{max} = 5$ and $\omega_\text{max} = 7$. This was done to ensure that the PDE eigensolver converged to the correct ground state energy. We then compared these reference eigenvalues against those obtained using the VMC simulations; as we can clearly see from the table, the eigenvalues match to one significant digit. We also note the mild increase in elapsed time in VMC simulations as the number of rotors is increased, compared to the exponential increase in the PDE eigensolver time. This implies that using VMC, we will be able to solve much larger problems than using the PDE eigensolver.

\subsection{VMC State Parametrization}

\begin{table}
    \centering
    \begin{tabular}{c c c c c c c}
        \toprule
        \multirow{2}{*}{\# Hidden nodes} & \multirow{2}{*}{Metric} & \multicolumn{4}{c}{SGD Step} & \multirow{2}{*}{Average Energy}                \\
        \cmidrule(lr){3-6}
                             &             & 100                  & 500                  & 5000                 & 9999                                                           \\
        \midrule
        \multirow{2}{*}{20}  & std.\ dev.  & $4.4 \times 10^{-1}$ & $1.9 \times 10^{-1}$ & $1.7 \times 10^{-1}$ & $1.7 \times 10^{-1}$ & \multirow{2}{*}{$1.4489 \times 10^{1}$} \\
                             & grad.\ norm & $1.1 \times 10^{-2}$ & $7.2 \times 10^{-4}$ & $1.3 \times 10^{-3}$ & $4.4 \times 10^{-3}$                                           \\
        \cmidrule(lr){1-7}
        \multirow{2}{*}{40}  & std.\ dev.  & $4.1 \times 10^{-1}$ & $1.8 \times 10^{-1}$ & $1.7 \times 10^{-1}$ & $1.5 \times 10^{-1}$ & \multirow{2}{*}{$1.4488 \times 10^{1}$} \\
                             & grad.\ norm & $1.3 \times 10^{-2}$ & $3.7 \times 10^{-3}$ & $6.5 \times 10^{-3}$ & $3.5 \times 10^{-3}$                                           \\
        \cmidrule(lr){1-7}
        \multirow{2}{*}{60}  & std.\ dev.  & $3.6 \times 10^{-1}$ & $2.0 \times 10^{-1}$ & $1.6 \times 10^{-1}$ & $1.6 \times 10^{-1}$ & \multirow{2}{*}{$1.4495 \times 10^{1}$} \\
                             & grad.\ norm & $1.7 \times 10^{-2}$ & $3.9 \times 10^{-3}$ & $5.0 \times 10^{-3}$ & $2.7 \times 10^{-3}$                                           \\
        \cmidrule(lr){1-7}
        \multirow{2}{*}{80}  & std.\ dev.  & $4.1 \times 10^{-1}$ & $1.9 \times 10^{-1}$ & $1.5 \times 10^{-1}$ & $1.5 \times 10^{-1}$ & \multirow{2}{*}{$1.4493 \times 10^{1}$} \\
                             & grad.\ norm & $5.9 \times 10^{-3}$ & $2.2 \times 10^{-2}$ & $1.7 \times 10^{-3}$ & $2.4 \times 10^{-3}$                                           \\
        \cmidrule(lr){1-7}
        \multirow{2}{*}{100} & std.\ dev.  & $4.4 \times 10^{-1}$ & $1.9 \times 10^{-1}$ & $1.5 \times 10^{-1}$ & $1.5 \times 10^{-1}$ & \multirow{2}{*}{$1.4496 \times 10^{1}$} \\
                             & grad.\ norm & $5.0 \times 10^{-4}$ & $1.8 \times 10^{-4}$ & $6.6 \times 10^{-4}$ & $1.0 \times 10^{-5}$                                           \\
        \bottomrule
    \end{tabular}
    \caption{Evolution of the energy standard deviation and gradient norm in the VMC simulations with ten-rotor chain network as we vary the number of hidden units. The VMC eigensolvers were run for $10000$ stochastic gradient descent steps with learning rate of $10^{-2}$ and $120000$ Metropolis-Hastings samples per gradient descent step ($20000$ samples were discarded as burn-in and then every 100th sample were cherry-picked). The last column reports the converged energy for each of these models.}
    \label{tab:ten-rotors-density}
\end{table}

The main advantage of the VMC approach over the PDE approach is the ability to parametrize the infinite-dimensional quantum rotor state using a finite-dimensional manifold of parameters. The dimensionality of the manifold is dependent on the number of hidden nodes; as the number of hidden nodes is increased, we are able to capture increasingly complicated quantum states. Note however, this also increases the number of parameters to learn (on top of increasing the computational complexity), and this may lead to poor performance if the model is not trained for long enough. In Table~\ref{tab:ten-rotors-density}, we record the eigenvalue standard deviation and gradient norm of a ten-rotor chain quantum rotor system with various number of hidden RBM nodes during various stages of the training process. We note that after 100 stochastic gradient descent steps, the RBM model with 60 hidden nodes has the smallest standard deviation in the eigenvalue; however as the training continues, all of the models end up with comparable converged eigenvalue.

\section{Conclusion and Future Directions}

We introduced continuous-variable neural quantum states as a variational ansatz for finding the ground-states of quantum Hamiltonian operators on continuous manifolds. We demonstrated the ability of these neural states to converge to the minimal eigenpair of the rotor Hamiltonian by comparing the obtained eigenvalue against those obtained using a baseline PDE based eigensolver. We observed that the scalability of our variational solver increases far slower when compared to the PDE eigensolver as the number of dimensions increase.

The baseline PDE eigensolver introduced in this paper leverages simple techniques from scalable scientific computing algorithms. While the implementation supports distributed computing, allowing us to scale beyond the memory limits of a single node, it does not address the issues related to the curse of dimensionality. Tensor factorization techniques can be used to compress the quantum state of high-dimensional systems to enable memory efficient computing. In \cite{beylkin2002numerical, beylkin2005algorithms}, the authors use the canonical polyadic (CP) decomposition to find the ground state of a Hamiltonian similar to the one we consider and in \cite{oseledets2011tensor}, the author uses the tensor-train (TT)/matrix-product state (MPS) decomposition on the same problem. \cite{rakhuba2016calculating, veit2017using} use the TT decomposition to find the vibrational spectra and ground states of molecules.

A line of investigation promising improved scalability of the VMC is the exploitation of symmetries of the interaction graph $G=(\mathcal{V},\mathcal{E})$. Although simple convolutional architectures  are likely sufficient for square grid graphs with discrete translational symmetry, a detailed investigation of the interplay between the automorphism group of $G=(\mathcal{V},\mathcal{E})$ and the isometry group of the target space $(M,g)$ would be desirable.

If the first quantized approach of this paper can ultimately be made to scale, then it should become possible to analyze the quantum phase transitions corresponding to the Berezinskii–Kosterlitz–Thouless (BKT) transition via the quantum classical mapping. In two (Euclidean) dimensions the BKT topological phase transition is a well-known consequence of the proliferation of vortices. Can systematically improvable variational wavefunctions cast light on the nature of the corresponding quantum phase transition? Similar techniques adapted to variational real-time dynamics could potentially offer a window into many-body quantum chaos via disordered rotor models \cite{cheng2019chaos} or the models with negative curvature \cite{gubser2016nonlinear}.

In closing, the approximation properties of continuous-variable neural-network quantum states is poorly understood. Considerable effort has been expended in the search for exact representations of many-body quantum states using restricted Boltzmann machines \cite{huang2017neural,deng2017machine,carleo2018constructing,chen2018equivalence,rrapaj2021exact,pei2021compact} and it would be very interesting if similiar techniques can be adapted to continuous-variable lattice systems.

\section{Acknowledgements}

J.S.~thanks Di Luo and Tobias Osborne for discussion and encouragement. Authors gratefully acknowledge support from NSF under grant DMS-2038030. This research was supported in part through computational resources and services provided by the Advanced Research Computing (ARC) at the University of Michigan.

\appendix
\clearpage

\section{Nonlinear Sigma Model Regularization}\label{sec:sigma}

A nonlinear sigma model is a theory of maps from a spacetime manifold to the target Riemannian manifold $(M,g)$. In this appendix we focus on Minkowski space $\mathbb{R}^{n,1}$ which is topologically the product $\mathbb{R}^{n} \times \mathbb{R}$. The following discussion can be easily adapted to replace $\mathbb{R}^n$ with a manifold of finite spatial volume such as the  $n$-torus. The degrees of freedom of the theory consist of target-space-valued functions of spacetime, which we denote by $x^a = x^a(\sigma,\tau)$, where $a = 1,\ldots,\dim M$ indexes the implicit coordinates for the target manifold $M$. Classical trajectories of the sigma model satisfy a nonlinear wave equation obtained by extremizing the following functional,
\begin{align}\label{e:action}
    S[x]
    & =
    \frac{1}{2f } \int_{\mathbb{R}}
    {\rm d}\tau
    \int_{\mathbb{R}^n}
    {\rm d}^{n} \sigma
    \,
    g_{ab}(x) \left[  \dot{x}^a \dot{x}^b - (\vec{\nabla} x^a) \cdot (\vec{\nabla} x^b)\right] \enspace ,
\end{align}
where $f  > 0$ is an overall normalization.
Define the momentum density as the following functional derivative,
\begin{equation}
    \pi_a := \frac{\delta S}{\delta \dot{x}^a} = \frac{1}{f } g_{ab}(x) \dot{x}^b \enspace ,
\end{equation}
in terms of which the energy density $\mathcal{H} := \pi_a \dot{x}^a - \mathcal{L}$ is given by
\begin{align}
    \mathcal{H}
        = \frac{f }{2} g^{ab} \pi_a \pi_b + \frac{1}{2f } g_{ab} (\vec{\nabla} x^a) \cdot (\vec{\nabla} x^b) \enspace ,
\end{align}
where $\mathcal{L}$ is the Lagrangian density.
In the simplest Euler discretization, the total energy $E(\tau) = \int {\rm d}^n \sigma \, \mathcal{H}$ at time $\tau$ of a map $x^a = x^a(\sigma,\tau)$ can be obtained from the limiting procedure $E(\tau) = \lim_{a \to 0} E_a(\tau)$,
\begin{align}\label{e:lattticeenergy}
    E_a(\tau)
        & =
        \sum_{\sigma \in a\mathbb{Z}^n}
        \left[ \frac{f }{2a^n} g^{ab}\big(x(\sigma,\tau)\big) p_a(\sigma,\tau) p_b(\sigma,\tau) + \frac{a^{n-2}}{2f }  \sum_{\mu=1}^n g_{ab}\big(x(\sigma,\tau)\big) \delta_\mu x^a(\sigma,\tau) \delta_\mu x^b(\sigma,\tau)\right] \enspace ,
\end{align}
where $\delta_\mu x^a(\sigma,\tau): = x^a(\sigma+a \,\hat{e}_\mu,\tau) - x^a(\sigma,\tau)$ is the finite difference operator in the direction of the unit basis vector $\hat{e}_\mu \in \mathbb{R}^n$, and $p_a(\sigma,\tau) := a^d \, \pi_a(\sigma,\tau)$ is the momentum at vertex $\sigma \in a \mathbb{Z}^n$. The double summation runs over the edges of the cubic lattice graph. Observe that the  potential $g_{ab} \delta_\mu x^a \delta_\mu x^b$ corresponding to each edge of the graph is the quadratic approximation to the squared Riemannian distance $\dist^2(x(\sigma), x(\sigma + a \, \hat{e}_\mu))$, which provides an accurate approximation when the involved distances are much less than the radius of curvature of the target space $(M,g)$. The form of the energy \eqref{e:lattticeenergy} motivates a quantum Hamiltonian in which the quadratic potential is replaced by its nonlinear completion. It follows from the requirement of self-adjointness of the conjugate momentum operator, that the kinetic term of the quantum Hamiltonian is obtained by replacing $g^{ab} p_a p_b$ with minus the Laplace-Beltrami operator defined as follows
\begin{equation}
    \Delta = \frac{1}{\sqrt{g}} \frac{\partial}{\partial x^a}\left( \sqrt{g} g^{ab} \frac{\partial}{\partial x^b}\right) \enspace ,
\end{equation}
where $g=\det g_{ab}$. The resulting Hamiltonian is of the form \eqref{e:hamiltonian} with $V(r) = r^2$ and with uniform vertex and edge weights given by $h_i = f / a^n$ and $\beta_{ij} = a^{n-2}/(2f )$, respectively.

\section{Training the VMC} \label{sec:training}

The components of the variational derivatives are given by
\begin{align}
	\frac{\partial \log \psi(x)}{\partial \vec{c}_j}
		& = \vec{x}_j \enspace , \\
	\frac{\partial\log \psi(x)}{\partial \vec{b}_i}
		& = g(r_i ) \frac{\vec{y}_i}{r_i} \enspace , \\
	\frac{\partial\log \psi(x}{\partial w_{ij}}
		& = g(r_i ) \frac{\langle \vec{y}_i, \vec{x}_j \rangle}{r_i} \enspace ,
\end{align}
where $g(x) := I_1(x) / I_0(x)$ and $r_i := \Vert y_i \Vert$.
Each time a Metropolis update of the $j$th visible rotor occurs of the form $\vec{x}_j \gets \vec{x}'_j$, the variable $r_i$ is updated for all $i\in[m]$ according to the following rule:
\begin{align}
	\vec{y}_i
		& \gets y_i + w_{ij}(\vec{x}'_j-\vec{x}_j) \enspace , \\
	r_i
		& \gets \Vert \vec{y}_i \Vert \enspace .
\end{align}

In our quantum rotor setup, the vectors $\vec{x}_i \in S^1$ are parametrized by $\theta_i \in [-\pi, \pi)$ as $\vec{x}_i = (\cos\theta_i, \sin\theta_i)$. The Metropolis update is then simply updating this angle parameter
\begin{equation}
    \theta_i' = \theta_i + \delta_i, \quad \delta_i \sim \text{Uniform}(-a, a)
\end{equation}
while shifting $\theta_i'$ by multiples of $2\pi$ to keep it in the range $[-\pi, \pi)$. Here $a > 0$ is another hyperparameter for the training.

\section{Exact Energy for Jastrow Wavefunction}\label{sec:jastrow}

In this appendix we analytically compute the Rayleigh quotient $R_H(\psi_{\rm J})$ for the quantum rotor Hamiltonian on a linear graph,
\begin{equation}
    H = -\frac{h}{2} \sum_{i=1}^n \frac{\partial^2}{\partial \theta_i^2} + \sum_{i=1}^{n-1} \beta_i \big[2-2\cos(\theta_i - \theta_{i+1})\big] \enspace ,
\end{equation}
using a generalized Jastrow trial wavefunction of the form
\begin{equation}
    \psi_J(\theta_1,\ldots,\theta_n) = \frac{1}{\sqrt{2\pi}} \prod_{i=1}^{n-1} \varphi_i(\theta_i - \theta_{i+1}) \enspace ,
\end{equation}
where for each $1 \leq i \leq n-1$, the function $\varphi_i : \mathbb{R} \to \mathbb{R}$ is even, $2\pi$-periodic and satisfies the following normalization condition,
\begin{equation}
    \int_0^{2\pi} {\rm d}\theta \, \varphi_i(\theta)^2 = 1 \enspace ,
\end{equation}
which ensures overall normalization
\begin{equation}
    \int_{[0,2\pi]^n}{\rm d}\theta_1\cdots{\rm d}\theta_n \, \psi_{\rm J}(\theta_1,\ldots,\theta_n)^2 = 1 \enspace .
\end{equation}
If we define constant functions $\varphi_0(\theta) =\varphi_n(\theta) = 0$, then
\begin{align}
    \left\langle \psi_{\rm J} , \frac{\partial^2\psi_{\rm J}}{\partial \theta_i^2} \right\rangle
    & := \int_{[0,2\pi]^n}{\rm d}\theta_1\cdots{\rm d}\theta_n \, \psi_{\rm J}(\theta_1,\ldots,\theta_n) \frac{\partial^2}{\partial \theta_i^2} \psi_{\rm J}(\theta_1,\ldots,\theta_n) \enspace , \\
    & = \int_0^{2\pi} {\rm d}\theta \left[\varphi_{i-1}(\theta) \varphi_{i-1}''(\theta) + \varphi_i(\theta) \varphi_i''(\theta)\right] - 2 \left[\int_0^{2\pi} {\rm d}\theta \varphi_{i-1}(\theta)\varphi_{i-1}'(\theta)\right]\left[\int_0^{2\pi} {\rm d}\bar{\theta} \, \varphi_{i}(\bar{\theta})\varphi_{i}'(\bar{\theta})\right] \enspace , \\
    & = \int_0^{2\pi} {\rm d}\theta \left[\varphi_{i-1}(\theta) \varphi_{i-1}''(\theta) + \varphi_i(\theta) \varphi_i''(\theta)\right] \enspace ,
\end{align}
where in the last equality we used the fact that $\varphi_i'$ is odd since $\varphi_i$ is even. Similarly,
\begin{equation}
    \langle \psi_{\rm J}, \cos(\theta_i - \theta_{i+1}) \psi_{\rm J} \rangle
        = \int_{0}^{2\pi} \varphi_i(\theta) \cos\theta \enspace ,
\end{equation}
and thus the total energy in the state $\psi_{\rm J}$ is
\begin{equation}
    \langle \psi_{\rm J}, H \psi_{\rm J} \rangle = \sum_{i=1}^{n-1} \left[2\beta_i-\int_0^{2\pi} {\rm d}\theta \, \varphi_{i}(\theta) \big( h\, \varphi_i''(\theta) + 2\beta_i \cos\theta\big) \right] \enspace .
\end{equation}
In the particular case of Eq.~\eqref{e:jastrow} of the main text we have,
\begin{equation}
    \varphi_i(\theta) = \frac{\exp(w_i \cos\theta)}{\sqrt{2\pi I_0(2w_i)}} \enspace ,
\end{equation}
and the energy is found to be
\begin{equation}
    \langle \psi_{\rm J}, H \psi_{\rm J} \rangle = \sum_{i=1}^{n-1}\left[2\beta_i + \frac{I_1(2w_i)}{I_0(2w_i)}\left(\frac{h}{2}w_i - 2 \beta_i\right) \right] \enspace .
\end{equation}

\section{Preconditioned Fourier Eigensolver} \label{sec:Fourier}

Here, we provide details on the Fourier-based spectral method for solving \eqref{eq:schrodinger}.
Substituting the Fourier expansion \eqref{eq:fourier} in the eigenvalue problem \eqref{eq:schrodinger}, multiplying both sides by $\exp(i \vb*{\omega} \cdot \vb*{\theta})$ and integrating over the domain $[0, 2\pi)^n$, we obtain
\begin{equation}
    \label{eq:fourier-eigenstate}
    \frac{h}{2} \Vert \vb*{\omega} \Vert^2 \hat{\psi}(\vb*{\omega}) + \sum_{\{ i,j \} \in \CE} \beta_{ij} [2 \hat{\psi}(\vb*{\omega})  - \hat{\psi}(\vb*{\omega} + \vb*{e}_i - \vb*{e}_j) - \hat{\psi}(\vb*{\omega} - \vb*{e}_i + \vb*{e}_j)] = \lambda \hat{\psi}(\vb*{\omega}) \enspace , \quad\quad \vb*{\omega} \in \DZ^n \enspace ,
\end{equation}
where $\vb*{e}_i \in \DR^n$ is the vector of all zeros except at the $i$-th entry, which is $1$. Clearly, \eqref{eq:fourier-eigenstate} is an eigenvalue equation for the Fourier coefficients
\begin{equation}
    \hat{H} \hat{\psi}(\vb*{\omega}) = \lambda \hat{\psi}(\vb*{\omega}) \enspace ,
\end{equation}
where the operator $\hat{H}$ defined as
\begin{equation}
    \label{eq:schrodinger-fourier}
    \hat{H} \hat{\psi}(\vb*{\omega}) := \frac{h}{2} \Vert\vb*{\omega}\Vert^2 \hat{\psi}(\vb*{\omega}) + \sum_{\{ i,j \} \in \CE} \beta_{ij} [2 \hat{\psi}(\vb*{\omega})  - \hat{\psi}(\vb*{\omega} + \vb*{e}_i - \vb*{e}_j) - \hat{\psi}(\vb*{\omega} - \vb*{e}_i + \vb*{e}_j)] \enspace .
\end{equation}
Note that given $\psi \in L^2([0, 2\pi)^n)$ with periodic weak derivatives, the Fourier coefficients $\hat{\psi} \in L^2(\DZ^d)$ satisfy
\begin{equation}
    \label{eq:fourier-decay-condition}
    \sum_{\vb*{\omega} \in \DZ^d} \Vert \vb*{\omega} \Vert^2 \lvert \hat{\psi}(\vb*{\omega}) \rvert < \infty \enspace .
\end{equation}
This condition \eqref{eq:fourier-decay-condition} suggests that the Fourier coefficients decay rapidly as infinity norm of the the frequency vector $\vb*{\omega} = (\omega_1, \ldots, \omega_n)$ increases in size. Thus, setting the Fourier coefficients to zero outside a hypercube in the frequency space
\begin{equation}
    \hat{\psi}(\vb*{\omega}) = 0 \quad \text{if} \quad \Vert \vb*{\omega}\Vert_\infty := \max\{\abs{\omega_1}, \ldots, \abs{\omega_n}\} \geq \omega_\text{max} \enspace ,
\end{equation}
provides a sufficiently accurate approximation to the full wavefunction as long as the cut-off frequency $\omega_\text{max}$ is not too small. This assumption also allows us to represent the wavefunction as a $n$-dimensional $(2 \omega_\text{max} + 1) \times \cdots \times (2 \omega_\text{max} + 1)$ tensor $\hat{\uppsi}$ for computational purposes. Next, restriction $\hat{\mathsf{H}}$ of the Hamiltonian operator $\hat{H}$ to the frequency range $[-\omega_\text{max}, \omega_\text{max}]^n$ is sparse, with at most $2 \abs{\CE} + 1$ non-zero entries per row. This structure leads to efficient matrix-vector operations and minimal inter-node communications in distributed computing setups.

Following this restriction of the eigenstate and the Hamiltonian in the Fourier domain, we need to find the minimal eigenvalue-eigenvector pair corresponding to the linear system
\begin{equation}
    \label{eq:schrodinger-discrete}
    \hat{\mathsf{H}} \hat{\uppsi} = \lambda \hat{\uppsi} \enspace .
\end{equation}
This is equivalent to finding the maximal eigenvalue-eigenvector pair of the system
\begin{equation}
    \label{eq:schrodinger-discrete-inverse}
    \hat{\mathsf{H}}' \hat{\uppsi} = \lambda' \hat{\uppsi} \enspace , \quad \hat{\mathsf{H}}' = (\hat{\mathsf{H}} - \mu \mathsf{I})^{-1} \enspace , \quad \lambda' = (\lambda - \mu)^{-1} \enspace ,
\end{equation}
where $\mu$ is a lower bound on the eigenvalues of $\hat{H}$, implying $\hat{H} - \mu I$ is invertible. Such a lower bound is derived in Appendix~\ref{sc:lower-bound}.

Power iteration provides a simple method for computing the maximal eigenpair of \eqref{eq:schrodinger-discrete-inverse}: starting from an arbitrary initial state $\hat{\uppsi}_0$, we iterate
\begin{equation}
    \hat{\uppsi}_{k + 1} = \frac{\hat{\mathsf{H}}' \hat{\uppsi}_{k}}{\lVert\hat{\mathsf{H}}' \hat{\uppsi}_k\rVert} \enspace , \quad\quad \lambda_{k + 1} = \frac{\hat{\uppsi}_{k + 1} \cdot \hat{\mathsf{H}} \hat{\uppsi}_{k + 1}}{\hat{\uppsi}_{k + 1} \cdot \hat{\uppsi}_{k + 1}} \enspace , \quad\quad k \geq 0 \enspace .
\end{equation}
As long as the minimal eigenpair of $\hat{\mathsf{H}}$ is \emph{non}-degenerate, and the initial state $\hat{\uppsi}_0$ is not orthogonal to the minimal eigenstate, this (inverse) power iteration is guaranteed to converge to the correct answer \cite{trefethen1997numerical}.

Finally, note that computing $\hat{\upphi}_{k + 1} = \hat{\mathsf{H}}' \hat{\uppsi}_k$ in the iteration requires solving a linear system
\begin{equation}
    (\hat{\mathsf{H}} - \mu \mathsf{I}) \hat{\upphi}_{k + 1} = \hat{\uppsi}_k \enspace .
\end{equation}
This can be done by any iterative solvers such as conjugate gradient (CG) or generalized minimal residual (GMRES) methods. They only require a routine for computing $(\hat{\mathsf{H}} - \mu \mathsf{I}) \hat{\uppsi}$ for arbitrary vector $\hat{\uppsi}$ \cite{trefethen1997numerical}. We have already pointed out how the sparsity structure of $\hat{\mathsf{H}}$ can be used to compute the matrix-vector product efficiently.

In our implementation, we choose CG as the iterative linear solver. To accelerate the convergence, we utilize the diagonal preconditioner
\begin{equation}
    \label{eq:preconditioner}
    \hat{\mathsf{M}} \hat{\uppsi}(\vb*{\omega}) = \left(\frac{h}{2} \Vert \vb*{\omega} \Vert^2 + \sum_{\{ i,j \} \in \CE} 2 \beta_{ij} - \mu\right)^{\frac{1}{2}} \hat{\uppsi}(\vb*{\omega}) \enspace ,
\end{equation}
and solve the following linear system
\begin{equation}
    \left[\hat{\mathsf{M}}^{-1} (\hat{\mathsf{H}} - \mu \mathsf{I}) \hat{\mathsf{M}}^{-1}\right] \left[\hat{\mathsf{M}} \hat{\upphi}_{k + 1}\right] = \left[\hat{\mathsf{M}}^{-1} \hat{\uppsi}_k\right] \enspace .
\end{equation}

Algorithm~\ref{alg:inv-power-iter} describes the full inverse power iteration process for determining the ground state of the quantum rotor Hamiltonian.

\begin{algorithm}
    \caption{Inverse power iteration for ground state determination}
    \label{alg:inv-power-iter}
    \begin{algorithmic}[1]
        \Require {Laplacian prefactor $h$, edge weights $\qty{\beta_{ij} : (i, j) \in \mathcal{E}}$, frequency cutoff $\omega_\text{max}$, tolerances $\tau_\text{cg}, \tau_\text{inv}$}
        \Ensure {Ground state $\hat{\uppsi}$ and energy $\lambda$}
        \State {Construct Hamiltonian $\hat{\mathsf{H}}$ using \eqref{eq:schrodinger-fourier} and $\omega_\text{max}$}
        \State {Construct preconditioner $\hat{\mathsf{M}}$ using $\eqref{eq:preconditioner}$ and $\omega_\text{max}$}
        \State {Determine eigenvalue lower bound $\mu$ from \eqref{eq:eigenvalue-lower-bound}}
        \State {Initialize $\uppsi_0 \in \mathbb{R}^{(2 \omega_\text{max} + 1) \times \cdots \times (2 \omega_\text{max} + 1)}$ randomly and normalize}
        \State {Initialize $\lambda_0 \gets \hat{\uppsi}_0 \cdot \hat{\mathsf{H}} \hat{\uppsi}_0$}
        \State {Initialize $k \gets 0$}
        \Repeat
            \State {Solve $(\hat{\mathsf{H}} - \mu \mathsf{I}) \hat{\upphi}_{k + 1} = \hat{\uppsi}_k$ using CG iteration with preconditioner $\hat{\mathsf{M}}$ and tolerance $\tau_\text{cg}$}
            \State {Design new state $\hat{\uppsi}_{k + 1} \gets \hat{\upphi}_{k + 1} / \lVert\hat{\upphi}_{k + 1}\rVert_2$}
            \State {Determine eigenvalue $\lambda_{k + 1} \gets \hat{\uppsi}_{k + 1} \cdot \hat{\upphi}_{k + 1}$}
            \State {$k \gets k + 1$}
        \Until {$\abs{\lambda_k - \lambda_{k - 1}} < \tau_\text{inv}$}
        \State {$\hat{\uppsi} \gets \hat{\uppsi}_k$}
        \State {$\lambda \gets \lambda_k$}
    \end{algorithmic}
\end{algorithm}

\section{Lower Bound on Minimal Eigenvalue} \label{sc:lower-bound}

Here, we derive a simple lower bound on the magnitude of the eigenvalues of the rotor Hamiltonian $H$, and equivalently, the operator $\hat{H}$ in the Fourier basis:

\begin{theorem}
    All eigenvalues $\lambda$ of the Hamiltonian \eqref{eq:schrodinger} satisfy $\lambda \geq \mu$ with
    \begin{equation}
        \label{eq:eigenvalue-lower-bound}
        \mu = -4 \sum_{\{ i,j \} \in \CE} \abs{\beta_{ij}} \enspace .
    \end{equation}
\end{theorem}

\begin{proof}
    Let us denote
    \begin{equation}
        h_{ij}(\vb*{\theta}) = \beta_{ij} (2 - 2 \cos{(\theta_i - \theta_j)}) \implies \abs{h_{ij}(\vb*{\theta})} \leq 4 \abs{\beta_{ij}}
    \end{equation}
    so that we can write
    \begin{equation}
        H = -\frac{h}{2} \Delta + \sum_{\{ i,j \} \in \CE} h_{ij}(\vb*{\theta})
    \end{equation}
    It follows that
    \begin{equation}
        \langle \psi, H \psi \rangle = \frac{h}{2} \langle \psi, -\Delta \psi \rangle + \sum_{\{ i,j \} \in \CE} \langle \psi, h_{ij} \psi \rangle
    \end{equation}
    and we can compute
    \begin{equation}
        \langle \psi, -\Delta \psi \rangle = -\int \dd{\vb*{\theta}} \psi^*(\vb*{\theta}) \Delta \psi(\vb*{\theta}) = \int \dd{\vb*{\theta}} \nabla \psi^*(\vb*{\theta}) \cdot \nabla \psi(\vb*{\theta}) = \int \dd{\vb*{\theta}} \Vert \nabla \psi(\vb*{\theta}) \Vert^2 \geq 0
    \end{equation}
    where the integration is performed over the domain $[0, 2\pi)^n$ and the second equality follows from integration by parts (boundary terms vanish due to the periodic boundary conditions). Further
    \begin{equation}
        \begin{split}
            \abs{\langle \psi, h_{ij} \psi \rangle} = \abs{\int \dd{\vb*{\theta}} \psi^*(\vb*{\theta}) h_{ij}(\vb*{\theta}) \psi(\vb*{\theta})} \leq \int \dd{\vb*{\theta}} \abs{h_{ij}(\vb*{\theta})} \abs{\psi(\vb*{\theta})}^2 \leq 4 \abs{\beta_{ij}} \int \dd{\vb*{\theta}} \abs{\psi(\vb*{\theta})}^2 = 4 \abs{\beta_{ij}} \langle \psi, \psi \rangle
        \end{split}
    \end{equation}
    It follows that
    \begin{equation}
        \langle \psi, H \psi \rangle \geq \frac{h}{2} \cdot 0 - 4 \sum_{\{ i,j \} \in \CE}\abs{\beta_{ij}} \langle \psi, \psi \rangle \implies R_H(\psi) = \frac{\langle \psi, H \psi \rangle}{\langle \psi, \psi \rangle} \geq - 4 \sum_{\{ i,j \} \in \CE} \abs{\beta_{ij}}
    \end{equation}
    Clearly,
    \begin{equation}
        \mu = -4 \sum_{\{ i,j \} \in \CE} \abs{\beta_{ij}}
    \end{equation}
    serves as a lower bound for the eigenvalues of the Hamiltonian $H$.
\end{proof}

\bibliographystyle{unsrt}
\bibliography{references}

\end{document}